\documentclass[final]{siamltex}
\usepackage{cite}
\usepackage{algorithm}
\usepackage{algorithmic}
\usepackage{amsmath,amssymb}
\usepackage{bm}
\usepackage{booktabs}
\usepackage{epsfig, epstopdf}
\usepackage{float}
\usepackage{graphics}
\usepackage{graphicx}
\usepackage{multido}
\usepackage{psfrag}
\usepackage{pstricks}
\usepackage{pst-all}
\usepackage{pst-poly}
\usepackage{setspace}
\usepackage{subfigure}
\usepackage{xypic}
\usepackage{color}
\usepackage{mathdots}

 \newtheorem{example}[theorem]{Example}

\newcommand{\ra}[1]{\renewcommand{\arraystretch}{#1}}
\newcommand{\coord}[1]{\widehat{#1}}
\newcommand{\mypar}[1]{{\bf #1.}}

\newcommand{\C}{\mathbb{C}}
\newcommand{\A}{\mathcal{A}}
\newcommand{\B}{\mathcal{B}}
\newcommand{\M}{\mathcal{M}}
\newcommand{\N}{\mathcal{N}}
\newcommand{\Z}{\mathbb{Z}}

\newcommand{\mylist}[1]{{\big\{ #1 \big\}}}
\newcommand{\myList}[1]{{\Big\{ #1 \Big\}}}
\newcommand{\gen}[1]{{\big\langle #1 \big\rangle}}

\newcommand{\floor}[1]{\lfloor #1 \rfloor}

\newcommand{\PT}{\mathcal{P}}
\newcommand{\DFT}{\textnormal{DFT}}

\newcommand{\DCTI}{\textnormal{DCT-1}}
\newcommand{\DSTI}{\textnormal{DST-1}}
\newcommand{\DCTII}{\textnormal{DCT-2}}
\newcommand{\DSTII}{\textnormal{DST-2}}
\newcommand{\DCTIII}{\textnormal{DCT-3}}
\newcommand{\DSTIII}{\textnormal{DST-3}}
\newcommand{\DCTIV}{\textnormal{DCT-4}}
\newcommand{\DSTIV}{\textnormal{DST-4}}
\newcommand{\pol}[1]{{#1}'}

\title{Algebraic Signal Processing Theory:
Cooley-Tukey Type Algorithms for Polynomial Transforms Based on
Induction\thanks{This work was supported in part by NSF grant CCF-0634967.}}

\author{    Aliaksei Sandryhaila\thanks{Department of
            Electrical and Computer Engineering, Carnegie Mellon University, Pittsburgh, PA 15213
                (asandryh@andrew.cmu.edu, pueschel@ece.cmu.edu). }
            \and Jelena Kova\v{c}evi\'{c}
            \thanks{Departments of Biomedical Engineering and Electrical and Computer Engineering,
            Carnegie Mellon University, Pittsburgh, PA 15213 (jelenak@cmu.edu). }
        \and Markus P\"{u}schel\footnotemark[2]
        }

\begin{document}

\maketitle


\begin{abstract}
A {\em polynomial transform} is the multiplication of an input vector
$x\in\C^n$ by a matrix $\PT_{b,\alpha}\in\C^{n\times n},$ whose
$(k,\ell)$-th element is defined as $p_\ell(\alpha_k)$ for polynomials
$p_\ell(x)\in\C[x]$ from a list $b=\{p_0(x),\dots,p_{n-1}(x)\}$ and
sample points $\alpha_k\in\C$ from a list
$\alpha=\{\alpha_0,\dots,\alpha_{n-1}\}$. Such transforms find
applications in the areas of signal processing, data compression, and
function interpolation. Important examples include the discrete
Fourier and cosine transforms. In this paper we introduce a novel
technique to derive fast algorithms for polynomial transforms. The
technique uses the relationship between polynomial transforms and the
representation theory of polynomial algebras. Specifically, we derive
algorithms by decomposing the regular modules of these algebras as a
stepwise induction. As an application, we derive novel $O(n\log{n})$
general-radix algorithms for the discrete Fourier transform and the
discrete cosine transform of type 4.
\end{abstract}



\pagestyle{myheadings}
\thispagestyle{plain}
\markboth{A.~SANDRYHAILA, J.~KOVA\v{C}EVI\'{C},
  AND M.~P\"{U}SCHEL}{FAST POLYNOMIAL TRANSFORMS BASED ON INDUCTION}


\section{Introduction}

\subsection{Polynomial transforms}

Let $b=\mylist{p_0(x),\dots,p_{n-1}(x)}\subset\C[x]$ be a list\footnotemark\ of complex polynomials
that form a basis of the space of polynomials of degree less than $n,$
and let $\alpha=\mylist{ \alpha_0,\dots,\alpha_{n-1}}\subset\C$ be a list of
distinct complex sample points.
%
\footnotetext[1]{%
Hereafter, we view lists as ordered sets, i.e., without
duplicate elements.
}%
A \emph{polynomial transform} is the matrix-vector product
$\PT_{b,\alpha}x$, where $x\in\C^n$ and $\PT_{b,\alpha}$ is the
$n\times n$ matrix whose $(k,\ell)$-th element is defined as
$p_\ell(\alpha_k)$, $0\leq k,\ell < n$:
\begin{equation}
\label{eq:PT_matrix}
\PT_{b,\alpha} =
\begin{pmatrix}
p_0(\alpha_0) & p_1(\alpha_0) & \dots & p_{n-1}(\alpha_0) \\
p_0(\alpha_1) & p_1(\alpha_1) & \dots & p_{n-1}(\alpha_1) \\
\vdots & \vdots && \vdots \\
p_0(\alpha_{n-1}) & p_1(\alpha_{n-1}) & \dots & p_{n-1}(\alpha_{n-1})
\end{pmatrix}.
\end{equation}
By a slight abuse of notation, we also refer to $\PT_{b,\alpha}$
as a polynomial transform.

Polynomial transforms are known in the literature under
different names.  For example, in~\cite{Driscoll:97}
and~\cite{Potts:98}, the authors refer to $\PT_{b,\alpha}$ as a
discrete polynomial transform. In~\cite{Kailath:97}, the authors call
it a polynomial Vandermonde matrix. The most well-known example of a
polynomial transform is the discrete Fourier transform (DFT).

Polynomial transforms have a number of important applications.  For
example, they are used for interpolation and
approximation~\cite{Fuhrman:96}, solving differential
equations~\cite{Boyd:2001}, data compression and image
processing~\cite{Mallat:99,Mandyam:96,Martens:90,Martens:90a}, and the
DFT specifically is widely used for spectral analysis and fast
computation of correlation and convolution.

The origin and main motivation for our work lies in the
\emph{algebraic signal processing
theory}~\cite{Pueschel:05e,Pueschel:08, Pueschel:08a}.  This theory
identifies polynomial transforms as equivalent to (generalized)
Fourier transforms for shift-invariant 1-D signal models, and
establishes a connection between these transforms and the
representation theory of polynomial algebras. This connection has been
used to algebraically derive many known and new fast algorithms for
the DFT and discrete cosine and sine transforms (DCTs and DSTs)
\cite{Pueschel:08a,Pueschel:08b} extending early ideas by Nussbaumer
\cite{Nussbaumer:82}. All these algorithms are derived and represented
as factorizations of the transform matrix into a product of
structured matrices with low computational costs.

In this paper, we develop a new algebraic method for a polynomial
transform factorization. It is based on viewing the associated
polynomial algebra as a regular module and decomposing it into an
induction using a chosen subalgebra. This decomposition, performed in
steps, yields a factorization of the polynomial transform. If all
factors have sufficiently low computational costs, this factorization is a fast
algorithm.

Our method extends the approach in \cite{Pueschel:08a,Pueschel:08b} to
its most general form.  As an application, we derive novel fast
general-radix algorithms for the DFT and the DCT of type 4 that
require only $O(n\log{n})$ operations instead of $n^2.$

\subsection{Related Work}
Over the last decades, decompositions that lead to fast algorithms
have been studied for certain polynomial transforms. Among them, the
DFT is arguably the most famous and well-studied.  The discovery of
the Cooley-Tukey fast Fourier transform (FFT)
algorithm~\cite{Cooley:65}, which reduced the computational cost of
$\DFT_n$ to $O(n\log{n})$ operations, led to decades of research and
numerous FFTs (see~\cite{VanLoan:92,Tolimieri:97} and the
references therein).

Most other polynomial transforms of interest are related to the DFT
and form the class of trigonometric transforms, since their entries
are cosine and sine expressions. This class includes the DCT and the
discrete sine transform (DST) of various types, as well as the real
DFT and the discrete Hartley transform.  Fast algorithms with
$O(n\log{n})$ operations have been developed, for example
in~\cite{Bergland:68,Bracewell:84,Wang:84,Sorensen:85,Duhamel:86,Rao:90}.

A more general class of polynomial transforms that were studied are
those based on orthogonal
polynomials~\cite{Kailath:97,Driscoll:97,Potts:98}.  With the
exception of DCT and DST, which belong to this group of transforms,
the fast algorithms for this class reported in the literature require
$O(n\log^2{n})$ operation.

Among hundreds of publications on this topic, most derived fast
algorithms by clever, but often complicated manipulations of the
matrix coefficients. This method provides little insight into the
origin and the basic principles that account for the existence of
these algorithms.

Another thread of research that we refer to as an algebraic theory of
transform algorithms has uncovered these principles for a large class
of algorithms for trigonometric transforms
\cite{Pueschel:03a,Pueschel:08b,Voronenko:09}. The theory exploits the connection
between polynomial transforms and polynomial algebras and uses
algebraic techniques to derive algorithms. As a result, most existing
algorithms were identified as special cases of two basic theorems, the
derivation is greatly simplified, and new algorithms were found.

The origin of the algebraic approach is in
\cite{Nicholson:71,Nussbaumer:82,Auslander:84,Beth:84}, who recognized
that the $\DFT_n$ can be interpreted as a decomposition matrix for the
group algebra $\C[\Z_n],$ where $\Z_n$ is a cyclic group of order
$n$~\cite{Auslander:84, Beth:84}. Since $\C[\Z_n]$ is identical to
the polynomial algebra $\C[x]/(x^n-1)$, this decomposition is
\begin{equation}
\label{eq:DFT_decomp}
\C[\Z_n]\cong\C[x]/(x^n-1) \rightarrow
\C[x]/(x-\omega_n^0)  \oplus \dots \oplus \C[x]/(x-\omega_n^{n-1}).
\end{equation}
Algorithms are now derived by performing this decomposition in steps
and reading off the respective matrices, which in turn factorize the DFT.

The group point of view was then generalized to derive fast Fourier
transforms for group algebras $\C[G]$ for noncyclic finite groups
$G$~\cite{Beth:87, Clausen:88, Diaconis:90, Rockmore:94, Clausen:93}.
Some of them were based on the induction for group algebras, a
construction that is algebraically analogous to the method used in
this paper.

The polynomial algebra point of view was extended to derive and study
larger classes of
FFTs~\cite{Winograd:79,Johnson:85,Auslander:84a,Heideman:86,Nussbaumer:82}.

The extension to the algorithm derivation of the full class of
trigonometric transforms and a large class of algorithms was then
accomplished in \cite{Pueschel:08b,Voronenko:09} based on early ideas
from \cite{Steidl:91,Steidl:92}. Since all these algorithms are based
on two theorems that generalize and account for the original
Cooley-Tukey FFT, all the algorithms were called ``Cooley-Tukey
type.''  The close relation between transforms and algebra was fully
developed and explained in the algebraic signal processing
theory~\cite{Pueschel:08,Pueschel:08a}.

In this paper we generalize the main theorem from \cite{Pueschel:08b}
and hence the class of Cooley-Tuke type algorithms. Specifically,
following the discussion in~\cite{Pueschel:08b}, we rigorously
demonstrate in Chapter~\ref{sec:MainTheorem} that these algorithms can
be viewed as based on a special case of algebraic induction. Then we
generalize the construction method to its most general form and show that it
produces novel algorithms. As examples, we derive new general-radix algorithms
for the DFT and the DCT of type 4.

\section{Polynomial algebras and transforms}

In this section we discuss polynomial algebras and
demonstrate that their decomposition matrices are exactly
polynomial transforms. We assume that the reader
is familiar with the basic theory of algebras, modules, and
matrix representations, even though we strive for a self-contained
presentation in this paper. A good introduction to these topics
can be found in \cite{Dummit:03,Curtis:62,Fuhrman:96}
Below, we briefly review definitions and important properties.

A vector space that is also a ring is called an \emph{algebra}. In
this paper, we work with \emph{polynomial algebras} of the form
$\A=\C[x]/p(x)$. Elements of $\A$ are polynomials in $x$ that are
added and multiplied modulo $p(x)$.  We assume
$p(x)=\prod_{k=0}^{n-1}{(x-\alpha_k)}\in\C[x]$ is a polynomial of
degree $n$ and separable, i.e. $\alpha_k \neq \alpha_m$ for $k\neq m$.
$\A$ is a commutative algebra of dimension $n$ with a multiplicative identity.

A vector space $\M$ that permits a multiplication by elements of $\A,$ such that
\begin{equation*}
am\in\M \text{  for any  } a\in\A,\ m\in\M,
\end{equation*}
is called an \emph{$\A$-module}. The special case $\M=\A$ is called a
\emph{regular} module. A subvector space $\N\leq\M$ that is
also closed under the multiplication by elements of $\A$, is called
an \emph{$\A$-submodule} of $\M$. If $\M$ has only trivial submodules
(i.e., $\{0\}$ and itself), it is called \emph{irreducible}.

It follows from the Wedderburn theorem that a regular module $\M=\A$
can be decomposed into a direct sum of irreducible
$\A$-modules~\cite{Curtis:62,Dummit:03}.  This decomposition is
accomplished by the Chinese Remainder Theorem:
\begin{equation}
\label{eq:isomorphism}
\begin{array}{rrcl}
\Delta:&\M&\rightarrow&\bigoplus_{k=0}^{n-1}{\C[x]/(x-\alpha_k)}, \\
&s(x)&\mapsto&\begin{pmatrix} s(\alpha_0) & s(\alpha_1) & \dots & s(\alpha_{n-1})\end{pmatrix}^T.
\end{array}
\end{equation}

Suppose the basis of $\M$ is a list \ of polynomials $b=\mylist{p_0(x), \dots, p_{n-1}(x)}$,
and in each  $\C[x]/(x-\alpha_k)$ we choose the basis consisting of $1$.
Then the matrix that describes the isomorphism~\eqref{eq:isomorphism} is precisely
the polynomial transform shown in~\eqref{eq:PT_matrix} :
\begin{equation}
\label{eq:PT}
\PT_{b,\alpha} = \left[ p_\ell(\alpha_k) \right]_{0\leq k,\ell < n} .
\end{equation}
Namely, $s(x)=\sum_{\ell=0}^{n-1}{s_\ell p_\ell(x)}\in\M$ becomes, in
coordinate form, the column vector
\begin{equation*}
\coord{s(x)} = \begin{pmatrix}s_0 & s_1 & \dots & s_{n-1} \end{pmatrix}^T,
\end{equation*}
and $\Delta(s(x))$ in~\eqref{eq:isomorphism} can be computed as the matrix-vector product
\begin{equation}
\label{eq:FT}
\Delta(s(x)) = \PT_{b,\alpha} \cdot \coord{s(x)} .
\end{equation}

%
%

\begin{example}
\label{ex:DFT}
\end{example}
If $b=\mylist{1,x,\dots,x^{n-1}}$ is the standard basis, then
the polynomial transform~\eqref{eq:PT} is the Vandermonde matrix
\begin{equation}
\label{eq:Vandermonde}
\PT_{b,\alpha} = \left[ \alpha_k^\ell \right]_{0\leq k,\ell < n}.
\end{equation}
If, in addition, $p(x)=x^n-1$, then $\alpha_k=\omega_n^k$, where
$\omega_n=e^{-i\frac{2\pi}{n}}$ with $i=\sqrt{-1}$, and the polynomial
transform is precisely the discrete Fourier transform
\begin{equation}
\label{eq:DFT}
\DFT_n = \left[ \omega_n^{k\ell} \right]_{0\leq k,\ell < n}.
\end{equation}

\begin{example}
\label{ex:DCT3}
\end{example}
If $b=\mylist{ T_0(x),\dots, T_{n-1}(x)}$ is the basis consisting of the
Chebyshev polynomials of the first kind\footnotemark,
then the polynomial transform has the form
\begin{equation}
\label{eq:DCT3_example}
\PT_{b,\alpha} = \left[ T_\ell(\alpha_k) \right]_{0\leq k,\ell < n}.
\end{equation}
If, in addition, $p(x)=T_n(x)$, then
$\alpha_k=\cos{\frac{(2k+1)\pi}{2n}}$ (see Table~\ref{tab:chebyshev}),
and the polynomial transform is the discrete cosine transform of type
3~\cite{Rao:90}:
\begin{equation}
\label{eq:DCTIII}
\DCTIII_n = \left[ \cos{\frac{(2k+1)\ell\pi}{2n}} \right]_{0\leq k,\ell < n}.
\end{equation}

%
\footnotetext[2]{ %
Chebyshev polynomials $C_k$ are the polynomials that satisfy the
two-term recurrence $C_{k+1} = 2xC_k - C_{k-1}$~\cite{Mason:02}. Hence, the whole
sequence of polynomials is determined by $C_0$ and $C_1$.
By setting $x=\cos{\theta}$, Chebyshev polynomials can also be
expressed in their trigonometric closed form as functions of
$\theta$. These and other properties are
shown in Table~\ref{tab:chebyshev}.
}%

%
\begin{table*}
\centering \ra{1.3}
{\footnotesize
\begin{tabular}{@{}llllll@{}}
\toprule %
Kind & $C$ & $C_0,\ C_1$ & $C_n(\cos\theta)$  & Symmetry & Zeros ($0\leq k < n$)\\
\midrule %
$1^\text{st}$ & $T$ & $1,x$ & $\cos{(n\theta)}$ & $T_{-n}=T_n$ & $\cos{\tfrac{(2k+1)\pi}{2n}}$ \\
$2^\text{nd}$ & $U$ & $1,2x$ & $\frac{\sin{(n+1)\theta}}{\sin{\theta}}$ & $U_{-n}=-U_{n-2}$ & $\cos{\tfrac{(k+1)\pi}{n+1}}$ \\
$3^\text{rd}$ & $V$ & $1,2x-1$ & $\frac{\cos{(n+\frac{1}{2})\theta}}{\cos{\frac{\theta}{2}}}$ & $V_{-n}=V_{n-1} $ & $\cos{\tfrac{(2k+1)\pi}{2n+1}}$ \\
$4^\text{th}$ & $W$ & $1,2x+1$ & $\frac{\sin{(n+\frac{1}{2})\theta}}{\sin{\frac{\theta}{2}}}$ & $W_{-n}=-W_{n-1}$ & $\cos{\tfrac{(2k+2)\pi}{2n+1}}$ \\
\bottomrule
\end{tabular}
}
\caption{Chebyshev polynomials, their closed form $C_n(\cos\theta)$, symmetry, and zeros.}
\label{tab:chebyshev}
\end{table*}

\mypar{Scaled polynomial transforms}
The notion of a polynomial transform can be generalized
by allowing a different choice of a basis in the $\C[x]/(x-\alpha_k)$ in~\eqref{eq:isomorphism}.
Namely, if we choose the basis $\mylist{c_k}, c_k\in\C$ in each $\C[x]/(x-\alpha_k)$,
then~\eqref{eq:PT} becomes the \emph{scaled polynomial transform}
\begin{equation}
\label{eq:scaled_PT}
\pol{\PT}_{b,\alpha} =
\diag \left( \frac{1}{c_0}, \dots, \frac{1}{c_{n-1}} \right)
\cdot \PT_{b,\alpha},
\end{equation}
with $\PT_{b,\alpha}$ as defined in~\eqref{eq:PT}.

\begin{example}
\label{ex:DCT4}
\end{example}
Let $p(x)=T_n(x)$, and choose the basis $b=\mylist{V_0(x),\dots, V_{n-1}(x)}$ in $\M$,
where $V_\ell(x)$ is the $\ell$-th Chebyshev polynomial of the third kind.
If we choose $c_k=1/\cos{\frac{(k+1/2)\pi}{2n}}$, then the associated scaled polynomial transform is
the discrete cosine transform of type 4:
\begin{eqnarray*}
\PT_{b,\alpha} &=&
\diag_{0\leq k<n}\left( \cos{\frac{(k+1/2)\pi}{2n}} \right) \cdot
\left[ \frac{\cos{\frac{(k+1/2)(\ell+1/2)\pi}{n}}}{\cos{\frac{(k+1/2)\pi}{2n}}} \right]_{0\leq k,\ell <n} \\
&=& \left[ \cos{\frac{(k+1/2)(\ell+1/2)\pi}{n}} \right]_{0\leq k,\ell <n} \\
&=& \DCTIV_n.
\end{eqnarray*}
Note that all 16 types of discrete sine and cosine transforms are
scaled or unscaled polynomial transforms with bases consisting of
Chebyshev polynomials~\cite{Pueschel:08a}.


\section{Subalgebra and its structure}
\label{sec:Subalgebra}
In this section we discuss the structure of subalgebras of $\A=\C[x]/p(x)$.


\subsection{Definition}
Choose a polynomial $r(x)\in\A$, and consider the space of
polynomials in $r(x)$ with addition and multiplication performed
modulo $p(x)$:
\begin{equation}
\label{eq:subalgebra}
\B=\Big\{ \sum_{k\geq 0}{c_k r^k(x)} \mod p(x) \mid c_k\in\C \Big\},
\end{equation}
where all sums are finite.
We call $\B$ the \emph{subalgebra} of $\A$ \emph{generated} by $r(x)$
and write this as $\B = \gen{r(x)}\leq\A$.

\subsection{Structure}

Given $r(x)\in\A$, we first determine the dimension of $\B = \gen{r(x)}$.
Then we identify $\B$ with a polynomial algebra of the form $\C[y]/q(y)$ with a suitably
chosen polynomial $q(y)$.

Let $\alpha=\mylist{\alpha_0,\dots, \alpha_{n-1}}$ be the list of roots of $p(x)$.
The generator $r(x)$ maps $\alpha$ to the list
$\beta=\mylist{\beta_0,\dots,\beta_{m-1}},$ such that for each $\alpha_k\in\alpha$
there is a $\beta_j\in\beta,$ for which $r(\alpha_k)=\beta_j.$
Hence, $m\leq n$, since for some $k$ and $\ell$
we may have $r(\alpha_k)=r(\alpha_\ell)$.

\begin{theorem}
\label{thm:dimension_of_B}
The dimension of $\B=\gen{r(x)}$ is $\dim{\B}=m=|\beta|$.
\end{theorem}
\begin{proof}
Let $d=\dim{\B}$. Since $\B\leq\A,$ then $\dim\B\leq\dim\A$ and the polynomials
$\mylist{1,r(x),\dots,r^{n-1}(x)}$ span the entire $\B$.
From the isomorphism~\eqref{eq:isomorphism} we obtain
\begin{eqnarray*}
d & = & \rank
\begin{pmatrix}
\Delta(1), \Delta(r(x)),\dots,\Delta(r^{n-1}(x))
\end{pmatrix} \\
& = & \rank
\begin{bmatrix}
r^\ell(\alpha_k)
\end{bmatrix}_{0\leq k,\ell < n}.
\end{eqnarray*}
Since $r(\alpha_k)\in\beta$ and $|\beta|=m$, the above matrix has only $m$ different rows;
hence, $d \leq m$. On the other hand, it contains the full-rank $m\times m$ Vandermonde matrix
\begin{equation*}
\begin{bmatrix}
\beta_j^\ell
\end{bmatrix}_{0\leq j,\ell < m}
\end{equation*}
as a submatrix; hence, $d\geq m$. Thus, we conclude that $d=\dim{\B}=m$.
\end{proof}
Next, we identify $\B$ with a polynomial algebra.
\begin{theorem}
\label{thm:algebraic_structure_of_B}
The subalgebra $\B=\gen{r(x)}$ can be identified with
the polynomial algebra $\C[y]/q(y)$, where
$q(y)=\prod_{j=0}^{m-1}{(y-\beta_j)}$,
via the following canonical isomorphism of algebras:
\begin{equation}
\label{eq:B_isomorphism}
\begin{array}{rrcl}
\kappa:&\B&\rightarrow&\C[y]/q(y),\\
&r(x)&\mapsto&y.
\end{array}
\end{equation}
We indicate this canonical isomorphism as $\B\cong\C[y]/q(y).$
\end{theorem}
\begin{proof}
Observe that $\B$ and $\C[y]/q(y)$ have the same dimension $m$,
and $\kappa$ maps the generator $r(x)$ of $\B$ to the
generator $y$ of $\C[y]/q(y)$. Hence, it suffices to show
that $q(r(x))\equiv 0 \mod p(x)$ in $\B$. From~\eqref{eq:isomorphism}
we obtain
\begin{eqnarray*}
\Delta(q(r(x))) & = &
\begin{pmatrix}
q(r(\alpha_0)) & \dots & q(r(\alpha_{n-1})
\end{pmatrix}^T \\
& = &
\begin{pmatrix}
0 & \dots & 0
\end{pmatrix}^T,
\end{eqnarray*}
which implies that $q(r(x))\equiv 0 \mod p(x)$ in $\A$, and hence in $\B$.
\end{proof}

Let $c=\mylist{q_0(y),\dots,q_{m-1}(y)}$ be a basis of $\C[y]/q(y)$.
The polynomial transform~\eqref{eq:PT} that decomposes the regular
module $\C[y]/q(y)$ (and hence the regular $\B$-module $\B$) is given
by \eqref{eq:isomorphism} as
\begin{equation*}
\PT_{c,\beta} = \left[ q_\ell(\beta_j) \right]_{0\leq j,\ell < m} .
\end{equation*}

\begin{example}
\label{ex:Subalgebra}
\end{example}
Consider the polynomial algebra $\A=\C[x]/(x^4-1)$ with
$\alpha=\mylist{1,-i,-1,i}$.  The polynomial $r_1(x)=x^2$ generates
the subalgebra $\B_1=\gen{r_1(x)}\cong\C[y]/(y^2-1)$ of dimension 2,
since $r_1(x)$ maps $\alpha$ to $\beta=\mylist{1,-1}$.

The polynomial $r_2(x)=(x+x^{-1})/2 = (x+x^3)/2$ generates the subalgebra
$\B_2=\gen{r_2(x)}\cong\C[y]/(y^3-y)$ of dimension 3, since $r_2(x)$
maps $\alpha$ to $\beta=\mylist{1,0,-1}$.


\section{Module induction}
\label{sec:Induction}
In this section we introduce the concept of \emph{module induction},
which constructs an $\A$-module $\M$ from a $\B$-module $\N$ for a
subalgebra $\B\leq \A$.  We will show that every regular $\A$-module
is an induction, which is the basis of our technique
for polynomial transform decomposition.

\subsection{Induction}
Similar to the coset decomposition in group
theory~\cite{Dummit:03,Curtis:62}, we can decompose a polynomial
algebra $\A=\C[x]/p(x)$ using a subalgebra $\B$ and associated \emph{transversal}:
\begin{definition}[Transversal]
Let $\B\leq\A$ be a subalgebra of $\A$.
A \emph{transversal} of $\B$ in $\A$ is a list of polynomials
$T=\mylist{t_0(x),\dots,t_{L-1}(x)}\subset\A$, such that, as vector spaces,
\begin{equation}
\label{eq:transversal}
\A=\bigoplus_{\ell=0}^{L-1}{t_\ell(x)\B}
=t_0(x)\B\oplus\dots\oplus t_{L-1}(x)\B.
\end{equation}
\end{definition}

Later, in Theorem~\ref{thm:existence_transversal}, we establish necessary and
sufficient conditions for a list of polynomials to be a transversal of $\B$
in $\A$. In particular, for any $\B\leq\A$ there always exists a transversal.

Given a transversal of $\B$ in $\A$, we define the module induction,
which is analogous to the induction for group algebras
in~\cite{Curtis:62}.
\begin{definition}[Induction]
Let $\B\leq\A$ be a subalgebra of $\A$ with a transversal $T$ as
in~\eqref{eq:transversal}, and let $\N$ be a $\B$-module.  Then the
following construction is an $\A$-module:
\begin{equation}
\label{eq:module_induction}
\M = \bigoplus_{\ell=0}^{L-1}{t_\ell(x)\N},
\end{equation}
where the direct sum is again of vector spaces.  It is called the
\emph{induction} of the $\B$-module $\N$ with the transversal $T$ to
an $\A$-module.  We write this as $\M=\N \uparrow_T \A$.
\end{definition}

In this paper, we are primarily interested in regular modules.  These
are always inductions, as follows directly from~\eqref{eq:transversal}
and~\eqref{eq:module_induction}:
\begin{lemma}
Let $\B\leq\A$ with a transversal $T$. Then the regular module $\A$ is
an induction of the regular module $\B$:
\begin{equation}
\label{eq:regular_induction}
\A = \B \uparrow_T \A.
\end{equation}
\end{lemma}

\subsection{Structure of cosets}

We have established in~\eqref{eq:B_isomorphism} that the subalgebra
$\B\leq \A$, generated by $r(x)\in\A$, can be identified with a
polynomial algebra $\C[y]/q(y)$.  Next, we investigate the structure
of each $\B$-module $t_\ell(x)\B$ in the
induction~\eqref{eq:regular_induction}.

Consider a polynomial $t(x)\in\A$.
As in Theorem~\ref{thm:algebraic_structure_of_B}, let $r(x)$ map $\alpha$ to $\beta$,
and let $q(y)=\prod_{j=0}^{m-1}{(y-\beta_j)}$.
Further, let $\alpha'=\mylist{ \alpha_k\mid t(\alpha_k)\neq 0 }\subseteq \alpha$
be the sublist of $\alpha$ that consists of those $\alpha_k$ that are not roots of $t(x)$.
Finally, let $r(x)$ map $\alpha'$ to $\beta'\subseteq \beta$, and denote $|\beta'|=m'$.

\begin{theorem}
\label{thm:dimension_of_tB}
The dimension of $t(x)\B$ is $\dim{t(x)\B}=|\beta'|=m'$.
\end{theorem}
\begin{proof}
The proof is similar to that of Theorem~\ref{thm:dimension_of_B}.
The list of polynomials $\mylist{t(x),t(x)r(x),\dots,t(x)r^{n-1}(x)}$
generates $t(x)\B$ as a vector space.
Using the isomorphism $\Delta$ in~\eqref{eq:isomorphism} we obtain
\begin{eqnarray}
\nonumber
\dim{\Big(t(x)\B\Big)} & = & \rank
\begin{pmatrix}
\Delta(t(x)), \Delta(t(x)r(x)),\dots,\Delta(t(x)r^{n-1}(x))
\end{pmatrix} \\
\label{eq:matrix_tDelta}
& = & \rank
\begin{bmatrix}
t(\alpha_k)r^\ell(\alpha_k)
\end{bmatrix}_{0\leq k,\ell < n} \\
\nonumber
& = & \rank \Big(
\diag\Big( t(\alpha_k) \Big)_{0\leq k < n}
\cdot
\begin{bmatrix}
r^\ell(\alpha_k)
\end{bmatrix}_{0\leq k,\ell < n} \Big).
\end{eqnarray}

Theorem~\ref{thm:algebraic_structure_of_B} shows that $\begin{bmatrix}
r^\ell(\alpha_k)\end{bmatrix}_{0\leq k,\ell < n}$ has exactly
$m=|\beta|$ linearly independent rows of the form
\begin{equation*}
\begin{pmatrix}
1 & \beta_j & \beta_j^2 & \dots & \beta_j^{n-1}
\end{pmatrix}.
\end{equation*}

For each $\beta_j$, the above row contributes exactly 1 to the rank of the matrix~\eqref{eq:matrix_tDelta}
if and only if there exists $\alpha_k$ such that $t(\alpha_k)\neq 0$ and $r(\alpha_k)=\beta_j$.
Since there are exactly $|\beta'|=m'$ such values of $\beta_j$, we conclude that
$\dim{\big(t(x)\B\big)}=m'$.
\end{proof}

Next, we identify the $\B$-module $t(x)\B$ with a $\C[y]/q(y)$-module.

\begin{theorem}
\label{thm:algebraic_structure_of_tB}
The $\B$-module $t(x)\B$ can be identified with the $\C[y]/q(y)$-module $\C[y]/q'(y)$, where
$q'(y)=\prod_{\beta_j\in\beta'}{(y-\beta_j)}$, via the module isomorphism
\begin{equation}
\label{eq:tB_isomorphism}
\begin{array}{rrcl}
\eta:& t(x)\B & \rightarrow & \C[y]/q'(y) , \\
&t(x)r^k(x) & \mapsto & y^k .
\end{array}
\end{equation}
By a slight abuse of notation, we write $t(x)\B\cong\C[y]/q'(y)$.
This is an isomorphism of modules and should not be confused with the
isomorphism of algebras in Theorem~\ref{thm:algebraic_structure_of_B}.
\end{theorem}
\begin{proof}
It follows from Theorem~\ref{thm:dimension_of_tB} that
$\mylist{t(x),t(x)r(x),\dots,t(x)r^{m'-1}(x)}$ is a basis of $t(x)\B$,
viewed as a vector space. On the other hand,
$\mylist{1,y,\dots,y^{m'-1}}$ is obviously a basis of $\C[y]/q'(y),$
also viewed as a vector space.  Hence, $\eta$
in~\eqref{eq:tB_isomorphism} is a bijective linear mapping between
$t(x)\B$ and $\C[y]/q'(y)$.

In order for $\eta$ to be an isomorphism of modules, it must also be a
module homomorphism---it must preserve the addition and multiplication
in $t(x)\B$ and $\C[y]/q'(y)$.  Namely, for $h(x)\in\B$ and
$u(x),v(x)\in t(x)\B$, the following conditions must hold:
\begin{eqnarray*}
\eta\big( u(x)+v(x) \big) &=& \eta(u(x)) + \eta(v(x)),\\
\eta\big( h(x)v(x) \big) &=& \kappa(h(x)) \cdot \eta(v(x)).
\end{eqnarray*}
The first condition is trivial. To show that the second condition
holds, let $h(x)=\sum_{k=0}^{m-1}{h_k r^k(x)}\in\B$ and
$v(x)=\sum_{j=0}^{m'-1}{v_j t(x)r^j(x)}\in t(x)\B$.  Then
\begin{eqnarray*}
\eta\big( h(x)v(x) \big) & = & \eta\big( \sum_{j=0}^{m+m'-2}{ \sum_{k=0}^j{h_k v_{j-k}} t(x)r^j(x)} \big)
= \sum_{j=0}^{m+m'-2}{ \sum_{k=0}^j{h_k v_{j-k}} y^j} \\
&=& \sum_{k=0}^{m-1}{h_k y^k} \cdot \sum_{j=0}^{m'-1}{v_j y^j}
= \kappa(h(x)) \cdot \eta(v(x)).
\end{eqnarray*}
Hence, $\eta$ is a module isomorphism.
\end{proof}

Note that, depending on $t(x)$, the dimension of $t(x)\B$ may be smaller than
the dimension of $\B$: $m' \leq m$. This effect is called \emph{annihilation}.

Also, the definition of $\eta$ in~\eqref{eq:tB_isomorphism} assumes the
standard basis $\mylist{1,y,\dots,y^{m'-1}}$ in $\C[y]/q'(y).$
If another basis $\mylist{b_0(y),\dots,b_{m'-1}(y)}$ were desired,
the corresponding basis in $t(x)\B$ would be
$\mylist{t(x)b_0(r(x)),\dots,t(x)b_{m'-1}(r(x))}.$

As a consequence of Theorem~\ref{thm:algebraic_structure_of_tB} and
the above discussion, decomposing the $\B$-module $t(x)\B$ with basis
$\mylist{t(x)q_0(r(x)),\dots,t(x)q_{m'-1}(r(x))}$ is the same as
decomposing the $\C[y]/q(y)$-module $\C[y]/q'(y)$ with basis
$c=\mylist{q_0(y),\dots,q_{m'-1}(y)}$. The decomposition matrix
is the same as for the regular module $\C[y]/q'(y)$ with the same basis,
namely
\begin{equation}
\label{eq:PT_tB}
\PT_{c,\beta'} = \left[ q_\ell(\beta_j) \right]_{0\leq j,\ell < m'} .
\end{equation}

\subsection{Existence of a transversal}

Consider $T=\mylist{ t_0(x),\dots,t_{L-1}(x)} \subset\A,$
and let $\dim{\big(t_\ell(x)\B\big)}=m_\ell$ for $0\leq \ell < L.$
Then $\mylist{ t_\ell(x),t_\ell(x)r(x),\dots,t_\ell(x)r^{m_\ell-1}(x)}$ is a basis
of $t_\ell(x)\B,$ as follows from Theorem~\ref{thm:dimension_of_tB}.
Hence, $T$ satisfies \eqref{eq:transversal}
if and only if $m_0+\dots+m_{L-1}=n$ and the concatenation of bases
\begin{equation}
\label{eq:bases_concatenation}
b' = \bigcup_{\ell=0}^{L-1}{ \myList{t_\ell(x),\dots,t_\ell(x)r^{m_\ell-1}(x) } }
\end{equation}
is a basis in $\A$. The following theorem states this condition in a matrix form.

\begin{theorem}
\label{thm:existence_transversal}
Using previous notation, $T$ is a transversal if and only if the following is a full-rank $n\times n$ matrix:
\begin{equation}
\label{eq:transversal_condition_matrix}
M' =
\begin{pmatrix}
D_0B_0 &|& D_1 B_1 &|& \dots &|&  D_{L-1} B_{L-1}
\end{pmatrix},
\end{equation}
where $D_\ell=\diag\left( t_\ell(\alpha_k)\right)_{0\leq k<n}$,
and $B_\ell = \left[ r^j(\alpha_k) \right]_{0\leq k < n, 0\leq j<m_\ell}$.
\end{theorem}
\\
\begin{proof}
The proof is similar to the proofs of Theorems~\ref{thm:dimension_of_B} and~\ref{thm:dimension_of_tB}.
Observe that the $k$-th element of $b'$ in ~\eqref{eq:bases_concatenation}
is mapped to the $k$-th column of $M'$ in ~\eqref{eq:transversal_condition_matrix}
by the isomorphism $\Delta$ in~\eqref{eq:isomorphism}. Hence, $b'$ is a basis in $\A$
if and only if $M'$ has exactly $n$ columns and $\rank{M'}=n$.
\end{proof}

It follows from Theorem~\ref{thm:existence_transversal} that for any algebra $\A$
and its subalgebra $\B$ there always exists a transversal. For example, we can choose
$T=\mylist{t_0(x),\dots,t_{n-1}(x)}$, where $t_\ell(\alpha_k)=0$ for $\ell\neq k$ and $t_\ell(\alpha_\ell)\neq 0$.
In this case $M'=\diag\left(t_\ell(\alpha_\ell)\right)_{0\leq\ell < n}$
in~\eqref{eq:transversal_condition_matrix}
is a full-rank diagonal matrix.

\begin{example}
\label{ex:Transversal}
\end{example}
Consider the subalgebras constructed in Example~\ref{ex:Subalgebra}.

For $\B_1=\gen{x^2}$ of dimension 2, we can choose the transversal $T=\mylist{1,x}$,
since $\mylist{1,x^2}\cup\mylist{x,x^3}$ is a basis for $\A$.
Since $x$ maps $\alpha$ to $\mylist{1,-i,-1,i}$, we have $\alpha'=\mylist{1,-i,-1,i}$ and
$\beta'=\mylist{1,-1}$. Hence, $q'(y)=(y-1)(y+1)$ and $x\B_1\cong\C[y]/(y^2-1)$ is of dimension 2.

For $\B_2=\gen{(x+x^{-1})/2}$ of dimension 3, we can choose the transversal
$T=\mylist{1,(x-x^{-1})/2}$, since the corresponding matrix
$$
M' =
\left(
\begin{array}{rrrr}
1 & 1 & 1 &  \\
1 &  &  & -i \\
1 & -1 & 1 &  \\
1 &  &  & i
\end{array}
\right)
$$
from~\eqref{eq:transversal_condition_matrix} has full rank.
Since $(x-x^{-1})/2$ maps $\alpha$ to $\mylist{0,-i,0,i}$,
we obtain $\alpha'=\mylist{-i,i}$, $\beta'=\mylist{0}$, and thus $q'(y)=y$.
Hence, $(x-x^{-1})/2\cdot \B_2\cong\C[y]/y$ is of dimension 1.

\section{Decomposition of polynomial transforms using induction}
\label{sec:MainTheorem}

In this section we use the induction~\eqref{eq:regular_induction} to
express the polynomial transform of $\A$ via the polynomial transforms of
each $t_\ell(x)\B\cong\C[y]/q'_\ell(y)$ in~\eqref{eq:transversal}.

As before, we consider $\A=\C[x]/p(x)$, where $p(x)=\prod_{k=0}^{n-1}{(x-\alpha_k)}$.
We view it as a regular $\A$-module with the chosen basis $b=\mylist{p_0(x),\dots,p_{n-1}(x)}$.

Let $\B=\gen{r(x)}\leq\A$ be a subalgebra generated by $r(x)\in\A,$
and $\B\cong\C[y]/q(y)$ according to Theorem~\ref{thm:algebraic_structure_of_B},
where $q(y)=\prod_{j=0}^{m-1}{(y-\beta_j)}$
and $\beta=\mylist{\beta_0,\dots,\beta_{m-1}}$.

Suppose $T=\mylist{t_0(x),\dots,t_{L-1}(x)}$ is a transversal of $\B$ in $\A$.
Let each $t_\ell(x)\B$ in~\eqref{eq:transversal}
be identified with a $\C[y]/q(y)$-module $\C[y]/q^{(\ell)}(y)$
according to Theorem~\ref{thm:algebraic_structure_of_tB},
where $q^{(\ell)}(y)=\prod_{\beta_j\in\beta^{(\ell)}}{(y-\beta_j)}$
and $m_\ell=|\beta^{(\ell)}|$.
The basis $b^{(\ell)}=\mylist{ b^{(\ell)}_0(y),\dots,b^{(\ell)}_{m_\ell-1}(y)}$ of $\C[y]/q^{(\ell)}(y)$
corresponds to the basis
$\mylist{ t_\ell(x)b^{(\ell)}_0(r(x)),\dots, t_\ell(x)b^{(\ell)}_{m_\ell-1}(r(x)) }$
of $t_\ell(x)\B$.
Hence, the corresponding polynomial transform~\eqref{eq:PT_tB} is $\PT_{b^{(\ell)},\beta^{(\ell)}}$.

\begin{theorem}
\label{thm:PT_via_induction}
Given the induction~\eqref{eq:regular_induction},
the polynomial transform $\PT_{b,\alpha}$ can be decomposed as
\begin{equation}
\label{eq:PT-via-induction}
\PT_{b,\alpha} =
\begin{pmatrix}
D_0M_0 &| &D_1M_1 &| &... &| &D_{L-1}M_{L-1}
\end{pmatrix}
\Big( \bigoplus_{\ell=0}^{L-1}\PT_{b^{(\ell)},\beta^{(\ell)}} \Big)
B.
\end{equation}
Here, $B$ is the base change matrix from the basis $b$ to the concatenation of bases
$$
\bigcup_{\ell=0}^{L-1}{ \mylist{ t_\ell(x)b^{(\ell)}_0(r(x)),\dots, t_\ell(x)b^{(\ell)}_{m_\ell-1}(r(x)) } }.
$$
Each $D_\ell=\diag\left( t_\ell(\alpha_k) \right)_{0\leq k<n}$ is a diagonal matrix.
Each $M_\ell$ is an $n\times m_\ell$ matrix whose $(k,j)$-th element
is $1$ if $r(\alpha_k)$ is equal to the $j$-th element of $\beta^{(\ell)}$, and
$0$ otherwise.
$\oplus$ denotes the direct sum of matrices:
$$
\bigoplus_{\ell=0}^{L-1}\PT_{b^{(\ell)},\beta^{(\ell)}}
=
\begin{pmatrix}
\PT_{b^{(0)},\beta^{(0)}} \\
& \PT_{b^{(1)},\beta^{(1)}}\\
&& \ddots \\
&&&\PT_{b^{(L-1)},\beta^{(L-1)}}
\end{pmatrix}.
$$
\end{theorem}
\begin{proof}
We prove the theorem for $L=2;$ that is, for $\A=t_0(x)\B\oplus t_1(x)\B$.
The proof for arbitrary $L$ is analogous.

Let  $\B\cong\C[y]/q(y)$ according to Theorem~\ref{thm:algebraic_structure_of_B},
where $q(y)=\prod_{j=0}^{m-1}{(y-\beta_j)}$
and $\beta=\mylist{\beta_0,\dots,\beta_{m-1}}$.
For $\ell\in\{0,1\}$, let $t_\ell(x)\B~\cong\C[y]/q^{(\ell)}(y)$
according to Theorem~\ref{thm:algebraic_structure_of_tB},
where $q^{(\ell)}(y)=\prod_{\beta_j\in\beta^{(\ell)}}{(y-\beta_j)}$
and $m_\ell=|\beta^{(\ell)}|$.
Let $b^{(\ell)}=\mylist{ b^{(\ell)}_0(y),\dots,b^{(\ell)}_{m_0-1}(y)}$
be a basis of $\C[y]/q^{(\ell)}(y)$.

Let $t_\ell(x)b^{(\ell)}(r(x)) = \mylist{ t_\ell(x)b^{(\ell)}_0(r(x)),\dots,t_\ell(x)b^{(\ell)}_{m_0-1}(r(x))}$.
As we established in Theorem~\ref{thm:existence_transversal},
$b'=t_0(x)b^{(0)}(r(x)) \bigcup t_1(x)b^{(1)}(r(x))$ is a basis of $\A$.
The original basis $b$ can be expressed in the new basis $b'$ as
$p_k(x)=\sum_{\ell=0}^{m_0-1}{B_{k,\ell}t_0(x) b^{(0)}_\ell(r(x))} +
\sum_{\ell=0}^{m_1-1}{C_{k,\ell}t_1(x)b^{(1)}_\ell(r(x))}$.
Hence, if $B$ is the base change matrix
from $b$ to $b'$, then
\begin{equation}
\label{eq:PT-new-basis}
\PT_{b,\alpha}=\PT_{b',\alpha}\cdot B.
\end{equation}
The $\ell$-th column of $B$ is
$\left(B_{0,\ell},\dots,B_{m_0-1,\ell},C_{0,\ell},\dots,C_{m_1-1,\ell}\right)^T.$

Next, observe that
\begin{equation}
\PT_{b',\alpha}=
\begin{pmatrix}
\PT_{t_0(x)b^{(0)}(r(x)),\alpha}  \mid  \PT_{t_1(x)b^{(1)}(r(x)),\alpha}
\end{pmatrix}.
\end{equation}
For each $\ell$, the $(k,j)$-th element of $\PT_{t_\ell(x)b^{(\ell)}(r(x)),\alpha}$
is $t_\ell(\alpha_k) b^{(\ell)}(r(\alpha_k))$. Hence,
\begin{equation}
\label{eq:PT-decomposed}
\PT_{t_\ell(x)b^{(\ell)}(r(x)),\alpha} = D_\ell \cdot M_\ell \cdot \PT_{b^{(\ell)},\beta^{(\ell)}},
\end{equation}
where $M_\ell$ is an $n\times m_\ell$ matrix whose $(k,j)$-th element is $1$ if $r(\alpha_k)$
equals to the $j$-th element of $\beta^{(\ell)}$, and $0$ otherwise;
and $D_\ell=\diag\Big(t_\ell(\alpha_k)\Big)_{0\leq k\leq n-1}$.

Hence, from (\ref{eq:PT-new-basis}-\ref{eq:PT-decomposed})  we obtain the desired decomposition:
\begin{equation}
\PT_{b,\alpha}
=
\begin{pmatrix}
D_0M_0 \mid D_1M_1
\end{pmatrix}
\cdot
\begin{pmatrix}
\PT_{b^{(0)},\beta^{(0)}} \oplus \PT_{b^{(1)},\beta^{(1)}}
\end{pmatrix}
\cdot B.
\end{equation}
\end{proof}

\begin{corollary}
\label{cor:Repeatition-of-column}
Consider the $n\times m$ matrix $M$ whose $(k,j)$-th element
is $1$ if $r(\alpha_k)=\beta_j$ and $0$ otherwise. Then
\begin{enumerate}
\item
$M$ contains exactly $n$ 1s and $n(m-1)$ 0s .

\item
Each matrix $M_\ell$ in Theorem~\ref{thm:PT_via_induction}
is a submatrix of $M$. It contains the $j$-th
column of $M$ if and only if $\beta_j\in\beta^{(\ell)}$.

\item
If the number of non-zero elements in the $j$-th column of $M$ is $c_j$,
then there are precisely $c_j$ matrices among $M_0,\dots,M_{L-1}$ that contain this column.
\end{enumerate}
\end{corollary}

\mypar{Discussion} The three factors in \eqref{eq:PT-via-induction}
correspond to the decomposition~\eqref{eq:isomorphism} of the regular
module $\A=\M=\C[x]/p(x)$ in three steps:

\textit{Step 1.}
$\A$ is represented as an induction~\eqref{eq:regular_induction} by
changing the basis in $\A$ to the concatenation of bases $b^{(\ell)}$ of $t_\ell(x)\B,$
using the base change matrix $B$.

\textit{Step 2.}
Each $t_\ell(x)\B$ is decomposed into a direct sum of irreducible $\B$-submodules,
using the corresponding polynomial transform $\PT_{b^{(\ell)},\beta^{(\ell)}}.$

\textit{Step 3.}
The resulting direct sum of irreducible $\B$-modules is decomposed into a
direct sum of irreducible $\A$-modules, using the matrix $M.$

The factorization~\eqref{eq:PT-via-induction} is a fast algorithm for
$\PT_{b,\alpha}$ if the matrices $B$ and $M$ have sufficiently low costs,
since the recursive nature of the second step allows for
repeated application of Theorem~\ref{thm:PT_via_induction}. We
illustrate this with two examples of novel algorithms derived using this theorem in
Section~\ref{sec:FastTransforms}.

\mypar{Special case: factorization of \boldmath{p(x)}} A special case
of Theorem~\ref{thm:PT_via_induction} has been derived
in~\cite{Pueschel:03a,Pueschel:08b}.  Namely, assume that
$\A=\C[x]/p(x),$ and we can decompose $p(x)=q(r(x)).$ Then
$\B=\gen{r(x)}\cong \C[y]/q(y)$, and any basis
$t=\mylist{1,t_1(x),\dots,t_{k-1}(x)}$ of $\C[x]/r(x)$ is a
transversal of $\B$ in $\A$. This leads to the following result.

\begin{corollary}
\label{cor:Decomposition}
Choose $c=\mylist{c_0(y),\dots,c_{m-1}(y)}$ as the basis of
$\C[y]/q(y)$.  Denote the roots of $r(x)-\beta_j$ as $\gamma^{(j)} =
\mylist{\gamma_0^{(j)},\dots,\gamma_{k-1}^{(j)}}$.  Notice that
$\bigcup_{j=0}^{m-1}\mylist{\gamma_0^{(j)},\dots,\gamma_{k-1}^{(j)}}$
is simply a permutation of $\mylist{\alpha_0,\dots,\alpha_{n-1}},$ and
denote the corresponding permutation matrix as $P$.  Then, the
polynomial transform decomposition~\eqref{eq:PT-via-induction} has the
form
\begin{equation}
\label{eq:decomposition}
\PT_{b,\alpha} = P^{-1} \Big(\bigoplus_{j=0}^{m-1}\PT_{t,\gamma^{(j)}} \Big)
L^n_m \Big(I_k \otimes \PT_{c,\beta} \Big) B.
\end{equation}
Here, $\otimes$ denotes the tensor product of matrices.
\end{corollary}

Corollary~\ref{cor:Decomposition} has been used to derive a large
class of fast algorithms for real and complex DFTs, and DCTs and
DSTs~\cite{Pueschel:03a,Pueschel:08b,Voronenko:09}.
Theorem~\ref{thm:PT_via_induction} further generalizes this approach,
and, as we show in the following example and in
Section~\ref{sec:FastTransforms}, also yields fast algorithms not based on
Corollary~\ref{cor:Decomposition}.

\begin{example}
\label{ex:Decomposition}
\end{example}
Consider the polynomial algebra $\A=\C[x]/(x^4-1)$ with basis $b=\mylist{1,x,x^2,x^3}.$ As we
showed in Example~\ref{ex:DFT}, the corresponding polynomial transform is $\PT_{b,\alpha}=\DFT_4.$

We continue from Example~\ref{ex:Transversal}.
First, consider $\B_1=\gen{x^2}$ and the induction
$\A = \B_1 \oplus x\B_1.$
Let us choose $b^{(0)}=\mylist{1,y}$ as the basis of $\C[y]/(y^2-1)\cong \B_1;$
it corresponds to the basis $\mylist{1,x^2}$ of $\B_1.$
We then choose $b^{(1)}=\mylist{1,y}$ as the basis of $\C[y]/(y^2-1)\cong x\B_1;$
it corresponds to the basis $\mylist{x,x^3}$ of $x\B_1.$
According to Theorem~\ref{thm:PT_via_induction},
$D_0=\diag\Big( 1,1,1,1 \Big),$ $D_1=\diag\Big( 1,-i,-1,i \Big),$
$$
M_0 = M_1 =
\left(
\begin{array}{rr}
1\\
&1\\
1\\
&1
\end{array}
\right),
\ \ \ \PT_{b^{(0)},\beta^{(0)}}=\PT_{b^{(1)},\beta^{(1)}}=
\left(
\begin{array}{rr}
1&1\\
1&-1
\end{array}
\right)
=\DFT_2,
$$
and $B$ is the base change matrix from $\mylist{1,x,x^2,x^3}$ to
$\mylist{1,x^2}\cup\mylist{x,x^3}.$
Hence,
\begin{equation}
\label{eq:DFT4_example_1}
\DFT_4 =
\left(
\begin{array}{rrrr}
1&&1\\
&1&&-i\\
1&&-1\\
&1&&i
\end{array}
\right)
\left(
\begin{array}{rr}
\DFT_2\\
&\DFT_2
\end{array}
\right)
\left(
\begin{array}{rrrr}
1\\
&&1\\
&1\\
&&&1
\end{array}
\right).
\end{equation}
As we show in Section~\ref{sec:Known_fast_transforms}, \eqref{eq:DFT4_example_1} is
exactly the Cooley-Tukey FFT for $\DFT_4$~\cite{Cooley:65}.

Next, consider $\B_2=\gen{(x+x^{-1})/2}$ and the induction $\A=\B_2
\oplus (x-x^{-1})/2\cdot\B_2.$ Let us choose
$b^{(0)}=\mylist{T_0(y),T_1(y),T_2(y)}=\mylist{1,y,2y^2-1}$ as the
basis of $\C[y]/(y^3-y)\cong\B_2;$ it corresponds to the basis
$\mylist{1,(x+x^{-1})/2,(x^2+x^{-2})/2}$ of $\B_2.$ We then choose
$b^{(1)}=\mylist{1}$ as the basis of
$\C[y]/y\cong(x-x^{-1})/2\cdot\B_2;$ it corresponds to the basis
$\mylist{(x-x^{-1})/2}$ of $(x-x^{-1})/2\cdot\B_2.$ According to
Theorem~\ref{thm:PT_via_induction}, $D_0=\diag\Big( 1,1,1,1 \Big),$
$D_1=\diag\Big( 0,-i,0,i \Big),$
$\PT_{b^{(1)},\beta^{(1)}}=\left(1\right)=\DSTI_1,$
$$
M_0 =
\left(
\begin{array}{rrr}
1\\
&1\\
&&1\\
&1
\end{array}
\right),
M_1 =
\left(
\begin{array}{r}
\\
1\\
\\
1
\end{array}
\right),
\PT_{b^{(0)},\beta^{(0)}}=
\left(
\begin{array}{rrr}
1&1&1\\
1&&-1\\
1&-1&1
\end{array}
\right)
=\DCTI_3,\\
$$
and $B$ is the base change matrix from $\mylist{1,x,x^2,x^3}$ to
$\mylist{1,(x+x^{-1})/2,(x^2+x^{-2})/2}\cup\mylist{(x-x^{-1})/2}.$
Hence,
\begin{equation}
\label{eq:DFT4_example_2}
\DFT_4 =
\left(
\begin{array}{rrrr}
1\\
&1&&-i\\
&&1\\
&1&&i
\end{array}
\right)
\left(
\begin{array}{rr}
\DCTI_3\\
&\DSTI_1
\end{array}
\right)
\left(
\begin{array}{rrrr}
1\\
&1&&1\\
&&1\\
&1&&-1
\end{array}
\right).
\end{equation}
As we show in Section~\ref{sec:New_fast_transforms}, \eqref{eq:DFT4_example_2}
is the Britanak-Rao algorithm for $\DFT_4$~\cite{Britanak:99}.

\section{Fast Signal Transforms}
\label{sec:FastTransforms}
In this section we apply the module induction to the construction of
novel fast algorithms for trigonomatric transforms, which are the most
important polynomial transforms used in signal processing. The
efficient computation of these transforms is of crucial importance in
most applications, and makes straightforward computation using
$O(n^2)$ operations prohibitive. As mentioned in the introduction,
many $O(n\log n)$ algorithms have been derived for these transforms
(e.g.,~\cite{Bergland:68,Bracewell:84,Wang:84,Sorensen:85,Duhamel:86,Rao:90})
and the origin of these algorithms was revealed by the algebraic
approach in \cite{Pueschel:03a,Pueschel:08b,Voronenko:09}, which also
produced new algorithms.

In this paper, we complete this work through
Theorem~\ref{thm:PT_via_induction} and its application. Specifically,
we will derive two novel $O(n\log n)$ general-radix algorithm that could not be
obtained with the prior algebraic theory.

We will first briefly touch on the algebraic signal processing theory
to explain why these transforms are associated with polynomial
algebras.  Then we derive the Cooley-Tukey FFT as special case of
Theorem~\ref{thm:PT_via_induction}, which motivates why we call all
such algorithms ``Cooley-Tukey type.'' Then we derive the novel
algorithms, both of which generalize existing algorithms that had
no satisfying algebraic explanation before.

%

\subsection{Algebraic Signal Processing}

In~\cite{Pueschel:03a,Pueschel:08, Pueschel:08a, Pueschel:08b}, the
authors introduced an axiomatic approach to the signal processing called
the \emph{algebraic signal processing theory}. They observed that the
basic assumptions used in signal processing are equivalent to viewing
\emph{filters} as elements of an algebra $\A$, and \emph{signals} as
elements of an associated $\A$- module $\M$.  In particular,
in the shift-invariant signal processing of
finite discrete one-dimensional filters and signals
$\A=\M=\C[x]/p(x)$ is necessarily a polynomial algebra.
The choice of $\A=\M$, together with a bijective mapping
$\Phi$ that maps samples from $\C^n$ to signals in $\M$, defines
a \emph{signal model} $(\A,\M,\Phi)$.

The fundamental tool in signal processing is the \emph{Fourier
transform}, which computes the frequency content of a signal.  From
the algebraic point of view, the Fourier transform for a signal model
$(\A,\M,\Phi)$ is precisely the decomposition~\eqref{eq:isomorphism}.
It can be computed as a matrix-vector product~\eqref{eq:FT} with the
appropriate polynomial transform~\eqref{eq:PT}.

\subsection{Notation}
Hereafter, we use the following special matrices:

$I_n$ is the identity matrix of size $n$.

$J_n$ is the complimentary identity matrix of size $n$: its $(k,n-1-k)$-th element
is $1$ for $0\leq k<n,$ and $0$ otherwise.

$\textbf{1}_n = \begin{pmatrix}1 & 1 & \dots & 1\end{pmatrix}^T$ is a column vector of $n$ ones.

$Z_n$ is the $n\times n$ circular shift matrix:
\begin{equation*}
Z_n = \begin{pmatrix}
& 1 \\
I_{n-1}
\end{pmatrix}.
\end{equation*}

$L_k^n,$ where $k$ divides $n$, is an $n\times n$ permutation matrix that selects
elements of $0,1,\dots, n-1$ at the stride $k$; the corresponding
permutation is $ik+j \mapsto jm+i$, where $0\leq i< m$ and $0\leq j < k$.
The $(i,j)$-th element of $L_k^n$ is $1$ if $j=\floor{\frac{ik(n+1)}{n}} \mod n,$
and $0$ otherwise.

$K^n_k=(I_{m}\oplus J_{m}\oplus I_{m}\oplus\dots)L_{k}^n,$ where $k$ divides $n$,
is another permutation matrix.

$T^n_k = \diag{ \Big( w_n^{ij} \mid 0\leq i<k, 0\leq j<m \Big)}$,
where the index $i$ runs faster, and $n=km$, is a twiddle factor matrix
used in the Cooley-Tukey FFT.

Complimentary direct sum:
$$
\oslash_{j=0}^{m-1}A_j =
\begin{pmatrix}
&& A_0 \\
& \iddots \\
A_{m-1}
\end{pmatrix}.
$$

\subsection{Cooley-Tukey FFT}\label{sec:Known_fast_transforms}

We derive the general-radix Cooley-Tukey FFT using
Theorem~\ref{thm:PT_via_induction}.  As was shown in
\cite{Pueschel:03a}, Corollary~\ref{cor:Decomposition} is sufficient
in this case.

Consider $\A=\M=\C[x]/(x^n-1).$ Let $b=\mylist{1,x,\dots,x^{n-1}}$ be
the basis of $\M.$ As we showed in Example~\ref{ex:DFT}, the
corresponding polynomial transform is $\DFT_n$.  Assume $n=km$. Let
$r(x)=x^k,$ and $\B=\gen{r(x)}.$ Then $x^\ell\B \cong\C[y]/(y^m-1),$
for $\ell=0\dots k-1,$ and $\A=\oplus_{\ell=0}^{k-1}{x^\ell\B}.$
Choosing the same basis $b^{(\ell)}=\{1,y,\dots,y^{m-1}\}$ in each
$\C[y]/(y^m-1)\cong x^\ell\B$ yields
$\PT_{b^{(\ell)},\beta^{(\ell)}}=\DFT_m.$ By
Theorem~\ref{thm:PT_via_induction}, we obtain
\begin{eqnarray*}
\DFT_{km} &=& M\cdot (I_k\otimes\DFT_m)\cdot B.
\end{eqnarray*}
Here, $B = L^{km}_k$ and $M = \left( D_0M_0 | \dots | D_{k-1}M_0 \right),$
where $M_0=\textbf{1}_k \otimes I_m,$ and
$D_\ell = \diag{\Big(\omega_{km}^{\ell j}\Big)}_{0\leq j < {km} }$ for $0\leq \ell < k$.
Hence, we can rewrite
$$
M = L^{km}_k \left( I_m \otimes \DFT_k \right) T^{km}_k L^{km}_m.
$$
to obtain the well-known general-radix Cooley-Tukey FFT algorithm~\cite{Cooley:65,Pueschel:08b}:
\begin{eqnarray}
\nonumber
\DFT_{km} &=& L^{km}_k \left( I_m \otimes \DFT_k \right) T^{km}_k  L^{km}_m
\left( I_k \otimes \DFT_m \right) L^{km}_k \\
\label{eq:CT_DFT}
&=& L^{km}_k \left( I_m \otimes \DFT_k \right) T^{km}_k \left( \DFT_m \otimes I_k \right).
\end{eqnarray}

\subsection{New Fast Algorithms}
\label{sec:New_fast_transforms}
In this section, we derive novel fast general-radix algorithms for $\DFT$ and $\DCTIV$.
Each of them requires $O(n\log{n})$ operations.
To the best of our knowledge, these algorithms have not been
reported in the literature.

\mypar{General-radix Britanak-Rao FFT}
In~\cite{Britanak:99}, Britanak and Rao derived a fast algorithm
for $\DFT_{2m}$ that can be written as the factorization
$$
\DFT_{2m} =
X^{2m}_m \Big(I_{m} \oplus Z^{-1}_{m} \Big)
D^{2m}_m \Big( \DCTI_{m+1} \oplus \DSTI_{m-1}\Big) B^{2m}_m.
$$
Matrices $D^{2m}_m,$ $B^{2m}_m,$ and $X^{2m}_m$
are specified in~(\ref{eq:BritRao_B}-\ref{eq:BritRao_D}) by setting $k=1.$

In Appendix~\ref{sec:BritanakRaoDerivation}, we derive the following
general-radix version of this algorithm:
\begin{theorem}
\label{thm:BritanakRao}
\begin{eqnarray*}
\DFT_{2km} &=&
L_{k}^{2km} \Big(I_{2m}\otimes \DFT_k\Big) X^{2km}_m L_{2m}^{2km}
\Big(I_{m} \oplus Z^{-1}_{m} \oplus I_{2(k-1)m}\Big) D^{2km}_m \\
&& \times \Big( \DCTI_{m+1} \oplus \DSTI_{m-1}
\oplus I_{k-1} \otimes (\DCTII_m\oplus \DSTII_m) \Big) B^{2km}_m.
\end{eqnarray*}
Here, $D^{2km}_m$ is a diagonal matrix,
and $B^{2km}_m$ and $X^{2km}_m$ are $2$-sparse matrices
(that is, with each row containing only two non-zero entries)
specified in~(\ref{eq:BritRao_B}-\ref{eq:BritRao_D}).

This factorization is obtained by inducing a subalgebra $\B=\gen{(x^k+x^{-k})/2}$
of an algebra $\A=\C[x]/(x^{2km}-1)$ with transversal
$t_0(x)=1,$ $t_1(x)=(x^k-x^{-k})/2,$ $t_{2j}(x)=x^j(x^k+1)/2,$
and $t_{2j+1}(x)=x^j(x^k-1)/2$ for $1\leq j < k.$
\end{theorem}

$\DFT_k$ requires $O(k\log k)$ operations; $\DCTI_{m+1},$ $\DSTI_{m-1},$ $\DCTII_m,$
and $\DSTII_m$ require $O(m\log m)$ operations each~\cite{Pueschel:03a,Pueschel:08b}.
$D^{2km}_m$ requires $n=2km$ operations and
$B^{2km}_m$ and $X^{2km}_m$ each require $3n$ operations.
Hence, the algorithm for $\DFT_{n}$ in Theorem~\ref{thm:BritanakRao}
requires $O(n\log n)$ operations.

\mypar{General-radix Wang algorithm for DCT-4}
In~\cite{Wang:84}, Wang derived a fast algorithm
for $\DCTIV_{2m}$ that can be written as the factorization
\begin{eqnarray*}
\DCTIV_{2m} &=&
K^{2m}_2
\cdot
\bigoplus_{j=0}^{m-1}
\begin{pmatrix}
\cos\frac{2m-2j-1}{8m}\pi &   (-1)^j\cos\frac{2j+1-2m}{8m}\pi \\
\cos\frac{2j+1-2m}{8m}\pi &   (-1)^{j+1}\cos\frac{2m-2j-1}{8m}\pi
\end{pmatrix} \\
&& \times (\DCTIII_m \otimes I_2) (K^{2m}_2)^T
\cdot
\begin{pmatrix}
1\\
&L_2^{2(m-1)}\cdot I_{m-1}\otimes \DFT_2\\
&&1
\end{pmatrix}.
\end{eqnarray*}

In Appendix~\ref{sec:WangDerivation}, we derive the following
general-radix version of this algorithm:
\begin{theorem}
\label{thm:Wang}
\begin{eqnarray*}
\DCTIV_{2km} &=&
K^{2km}_k (K^{2m}_2 \otimes \DCTIV_k) Y^{2km}_m  \cdot (\DCTIII_m \otimes L^{2k}_2) (K^n_{2k})^T \\
&& \times I_k \otimes
\begin{pmatrix}
1\\
&L_2^{2(m-1)}\cdot I_{m-1}\otimes \DFT_2\\
&&1
\end{pmatrix}
(K^{2km}_{2m})^T .
\end{eqnarray*}
Here, $Y^{2km}_m$ is a $2$-sparse matrix specified in~\eqref{eq:Wang_Y}.

This factorization is obtained by inducing a subalgebra $\B=\gen{T_{2k}(x)}$
of an algebra $\A=\C[x]/T_{2km}(x)$ with transversal $t_{2j}(x)=V_j(x)$
and $t_{2j+1}(x)=W_j(x)(V_{2k-1}(x)-V_{2k}(x))/2$ for $0\leq j < k.$
\end{theorem}

$\DCTIV_k$ requires $O(k\log k)$ operations, and
$\DCTIII_m$ requires $O(m\log m)$ operations~\cite{Pueschel:03a,Pueschel:08b}.
$Y^{2km}_m$ requires $3n$ operations, where $n=2km.$
Hence, the algorithm for $\DCTIV_{n}$ in Theorem~\ref{thm:Wang}
requires $O(n\log n)$ operations.


\section{Conclusion}
We have introduced a new approach to the factorization of polynomial
transforms $\PT_{b,\alpha}$ based on the decomposition of the
underlying regular module $\A=\M=\C[x]/p(x)$ into an induction. This
approach is in its most general form since the underlying
Theorem~\ref{thm:PT_via_induction} allows for arbitrary subalgebras.
Not every factorization based on this theorem yields a fast algorithm:
it depends on the computational costs of matrices $B$ and $M$ that occur
in its recursive application.

However, we have shown that the theorem produces at least two novel
general-radix algorithms for the DFT and a DCT. Both algorithms cannot
be obtained using the prior Corollary~\ref{cor:Decomposition}. In addition,
both generalize algorithms from the literature, which now become
the special cases of radix 2.

Equally important, we make another step towards a complete algebraic
theory of fast algorithms for polynomial transforms.


\mypar{Future work} In addition to the DFT, DCT, and DST, other
polynomial transforms have been studied. In particular, polynomial
transforms based on orthogonal polynomials have found applications in
such areas as function interpolation, data compression, and image
processing~\cite{Mandyam:96,Martens:90,Martens:90a}.  For practical
applications, fast algorithms for this class of polynomial transforms
are needed.  With the exception of DCT and DST, the fastest
algorithms, reported in the literature to date, require
$O(n\log^2{n})$ operations (in particular, more than $43n\log_2^2{n}$
for $n=2^k$)~\cite{Driscoll:97, Potts:98}. The question is whether our
approach can improve this bound for some or all of these transforms.


\appendix

\section{Proof of Theorem~\ref{thm:BritanakRao}}
\label{sec:BritanakRaoDerivation}
Consider $\A=\M=\C[x]/(x^{2km}-1)$, with basis $\mylist{ 1,x,\dots,x^{2km-1} }$
and $\alpha_k=\omega_{2km}^k$. The corresponding polynomial transform is $\DFT_{2km}$.

By Theorem~\ref{thm:algebraic_structure_of_B},
the polynomial $r(x)=(x^k+x^{-k})/2$ generates the subalgebra
$$
\B=\gen{r(x)} \cong \C[y]/2(y^2-1)U_{m-1}(y).
$$
If we choose
$\mylist{ T_\ell(y) }_{0\leq \ell < m+1}$
as the basis
, the polynomial transform is
$$
\Big[ T_\ell(\cos\frac{k\pi}{m}) \Big]_{0\leq k,\ell < m+1} = \DCTI_{m+1}.
$$

By Theorem~\ref{thm:algebraic_structure_of_tB},
the $\B$-module $(x^k-x^{-k})/2\cdot\B\cong\C[y]/U_{m-1}(y)$.
If we choose the basis
$\mylist{U_\ell(y)}_{0 \leq \ell < m-1}$,
then the polynomial transform is
\begin{equation}
\label{eq:polDSTI}
\Big[ U_\ell(\cos\frac{k\pi}{m}) \Big]_{0\leq k,\ell < m-1} =
\diag{\Big(
1/\sin\frac{(k+1)\pi}{m}
\Big)}_{0\leq k < m-2} \cdot \DSTI_{m-1}
=\pol{\DSTI}_{m-1}.
\end{equation}

Similarly,
the $\B$-module $x^j(x^k+1)/2\cdot\B\cong\C[y]/2(y-1)U_{m-1}(y)$ for any $1\leq j< k$.
If we choose the basis
$\mylist{V_\ell(y)}_{0 \leq \ell < m},$
then the polynomial transform is
\begin{equation}
\label{eq:polDCTII}
\Big[ V_\ell(\cos\frac{k\pi}{m}) \Big]_{0\leq k,\ell < m} =
\diag{ \Big(
1/\cos\frac{k\pi}{2m}
\Big)}_{0\leq k < m} \cdot \DCTII_m
=\pol{\DCTII}_m.
\end{equation}

Finally,
the $\B$-module $x^j(x^k-1)/2\cdot\B\cong\C[y]/2(y+1)U_{m-1}(y)$ for any $1\leq j< k$.
If we choose the basis
$\mylist{W_\ell(y)}_{0 \leq \ell < m},$
then the polynomial transform is
\begin{equation}
\label{eq:polDSTII}
\Big[ W_\ell(\cos\frac{(k+1)\pi}{m}) \Big]_{0\leq k,\ell < m} =
\diag{\Big( 1/\sin\frac{(k+1)\pi}{2m} \Big)}_{0\leq k < m} \cdot \DSTII_m
=\pol{\DSTII}_m.
\end{equation}

Using Theorem~\ref{thm:existence_transversal}, we can verify that
$t_0(x)=1,$ $t_1(x)=(x^k-x^{-k})/2,$ $t_{2j}(x)=x^j(x^k+1)/2,$ and $t_{2j+1}(x)=x^j(x^k-1)/2$ for $1\leq j < k,$
is a transversal of $\B$ in $\A$.
Hence, by Theorem~\ref{thm:PT_via_induction}, we obtain the factorization
\begin{eqnarray*}
\DFT_n &=& M  \Big( \DCTI_{m+1}\oplus \pol{\DSTI}_{m-1}
 \oplus I_{k-1}\otimes(\pol{\DCTII}_m \oplus \pol{\DSTII}_m) \Big)  B^{2km}_m.
\end{eqnarray*}
Here, $B^{2km}_m$ is the base change matrix from $\mylist{ x^\ell }_{0\leq \ell\leq n-1}$
to the concatenation of bases of $t_j(x)\B$, $0\leq j<2k,$ and by construction
\begin{equation}
\label{eq:BritRao_B}
B^{2km}_m =
\begin{pmatrix}
1\\
&I_{m-1}&&J_{m-1}\\
&&1\\
&I_{m-1}&&-J_{m-1}
\end{pmatrix}
\oplus
I_{k-1}
\otimes
\begin{pmatrix}
1&1\\
&&I_{m-1}&J_{m-1}\\
-1&1\\
&&I_{m-1}&-J_{m-1}\\
\end{pmatrix}
\cdot L_k^{2km}.
\end{equation}
$M$ is constructed as follows. Let
$$
M_0 =
\textbf{1}_k \otimes
\begin{pmatrix}
1\\
&I_{m-1}\\
&&1\\
&J_{m-1}
\end{pmatrix}.
$$
Let $M_0(j_0,\dots,j_\ell)$ be the subset of columns of $M_0$ with indices $j_0,\dots,j_\ell$;
and let $D_j=\diag{\Big(t_j(\alpha_i)\Big)}_{0\leq i< n},$ for $0\leq j<2k$.
Then
$$
M =
\begin{pmatrix}
D_0M_0\mid D_1M_1\mid D_2M_2\mid \dots \mid D_{2k-1}M_{2k-1}
\end{pmatrix},
$$
where $M_1=M_0(1,\dots,m-1)$; $M_{2j}=M_0(0,\dots,m-1)$ and $M_{2j+1}=M_0(1,\dots,m)$ for
$1\leq j<k$.
We can further rewrite $M$ as
$$
M = L_{k}^{2km} (I_{2m}\otimes \DFT_k) X^{2km}_m L_{2m}^{2km} (I_{m} \oplus Z^{-1}_{m} \oplus I_{2(k-1)m}).
$$
Here, matrix $X^{2km}_m$ has the structure
\begin{equation}
\label{eq:BritRao_X}
X^{2km}_m =
\begin{pmatrix}
I_k\\
& \oplus_{j=1}^{m-1} C_j & \oplus_{j=1}^{m-1} D_j \\
&&& F \\
& \oslash_{j=m+1}^{2m-1} C_j & \oslash_{j=1}^{m-1} D_j \\
\end{pmatrix},
\end{equation}
where
\begin{eqnarray*}
C_j &=&  1 \oplus \diag{\Big(\omega_{2km}^{j\ell}(\omega_{2m}^{j}+1)/2\Big)}_{1\leq \ell <k} , \\
D_j &=&  \Big( (\omega_{2m}^j-\omega_{2m}^{-j})/2 \Big)
\oplus \diag{\Big( \omega_{2km}^{j\ell}(\omega_{2m}^{j}-1)/2\Big)}_{1\leq \ell <k} , \\
F   &=& 1 \oplus \diag{\Big(-\omega_{2k}^j\Big)}_{1\leq j <k} .
\end{eqnarray*}

After the substitution of $\pol{\DSTI}_{m-1}$, $\pol{\DCTII}_m$, and $\pol{\DSTII}_m$ with
$\DSTI_{m-1}$, $\DCTII_m$, and $\DSTII_m$ using~(\ref{eq:polDSTI}-\ref{eq:polDSTII}),
and simplification, we obtain the factorization
\begin{eqnarray*}
\DFT_{2km} &=&
L_{k}^{2km} (I_{2m}\otimes \DFT_k) X^{2km}_m L_{2m}^{2km}
(I_{m} \oplus Z^{-1}_{m} \oplus I_{2(k-1)m}) D^{2km}_m \\
&& \cdot \Big( \DCTI_{m+1} \oplus \DSTI_{m-1}
\oplus I_{k-1} \otimes (\DCTII_m\oplus \DSTII_m) \Big)  B^{2km}_m,
\end{eqnarray*}
where $B^{2km}_m$ and $X^{2km}_m$ are defined in~\eqref{eq:BritRao_B} and
\eqref{eq:BritRao_X}, and
\begin{eqnarray}
\nonumber
D^{2km}_m &=&  I_{m+1} \oplus \diag{\Big(1/\sin\frac{(j+1)\pi}{m} \Big)}_{0\leq j< m-1} \\
\label{eq:BritRao_D}
&& \oplus I_{k-1} \otimes \Big(
\diag{\Big( 1/\cos\frac{j\pi}{2m} \Big)}_{0\leq j<m}
 \oplus \diag{\Big( 1/\sin\frac{(j+1)\pi}{2m} \Big)}_{0\leq j<m} \Big).
\end{eqnarray}

\section{Proof of Theorem~\ref{thm:Wang}}
\label{sec:WangDerivation}

Consider $\A=\M=\C[x]/2T_{2km}(x)$ with basis
$\mylist{ V_0(x),V_1(x),\dots,V_{2km-1}(x)}$. The corresponding polynomial
transform is
\begin{equation}
\label{eq:polDCTIV}
\diag{\Big(1/\cos\frac{(k+1/2)\pi}{4km}\Big)}_{0\leq k<2km} \cdot \DCTIV_{2km} = \pol{\DCTIV}_{2km}.
\end{equation}

By Theorem~\ref{thm:algebraic_structure_of_B},
the polynomial $r(x)=T_{2k}(x)$ generates the subalgebra
$$
\B=\gen{r(x)} \cong \C[y]/2T_m(y).
$$
By Theorem~\ref{thm:algebraic_structure_of_tB},
the $\B$-module $V_j(x)\B\cong\C[y]/2T_m(y)$ for any $0\leq j< k$.
If we choose the basis
$\mylist{V_\ell(y)}_{0 \leq \ell < m},$
then the polynomial transform is
$$
\Big[ T_\ell(\cos\frac{(k+1/2)\pi}{m}) \Big]_{0\leq k,\ell < m} = \DCTIII_m.
$$

Similarly,
the $\B$-module $W_j(x)(V_{2k-1}(x)-V_{2k}(x))/2\cdot\B\cong\C[y]/2T_{m}(y)$
for any $0\leq j< k$.
If we choose the basis
If we choose the basis
$\mylist{ U_\ell(y)}_{0 \leq \ell < m},$
then the polynomial transform is
\begin{equation}
\label{eq:polDSTIII}
\Big[ U_\ell(\cos\frac{(k+1/2)\pi}{m}) \Big]_{0\leq k,\ell < m-1} =
\diag{ \Big(1/\sin\frac{(k+1/2)\pi}{m} \Big)}_{0\leq k < m} \cdot \DSTIII_m
=\pol{\DSTIII}_m.
\end{equation}

We can verify using Theorem~\ref{thm:existence_transversal} that
$t_{2j}=V_j(x)$ and $t_{2j+1}=W_j(x)(V_{2k-1}(x)-V_{2k}(x))/2$ for $0\leq j < k,$
is a transversal of $\B$ in $\A$.
Hence, by Theorem~\ref{thm:PT_via_induction}, we obtain the decomposition
\begin{equation*}
\pol{\DCTIV}_{2km} = M
\left(
I_k\otimes(\pol{\DCTIII}_m \oplus \pol{\DSTIII}_m)
\right)
 B.
\end{equation*}
Here, $B$ is the base change matrix from $\mylist{ x^\ell }_{0\leq \ell\leq n-1}$
to the concatenation of bases of $t_j(x)\B$, $0\leq j<2k,$ and by construction
$$
B = I_k\otimes
\begin{pmatrix}
1\\
&L_2^{2(m-1)}\cdot I_{m-1}\otimes \DFT_2\\
&&1
\end{pmatrix}
(K^{2km}_{2m})^T.
$$
$M$ is constructed as follows. Let
$$
M_0 =
\textbf{1}_k \otimes
\begin{pmatrix}
I_m\\
J_m
\end{pmatrix}.
$$
Let $D_j=\diag{\Big(t_j(\alpha_i)\Big)}_{0\leq i< n}$ for $0\leq j<2k$.
Then
$$
M =
\begin{pmatrix}
D_0M_0\mid D_1M_0\mid D_2M_0\mid \dots \mid D_{2k-1}M_0
\end{pmatrix}.
$$

We can simplify matrix $M$. Let us introduce matrices
\begin{equation}
\label{eq:X_C4}
X_k^{(C4)}(r) = \begin{pmatrix}
c_0 &&& s_{k-1} \\
&\ddots & \iddots \\
&\iddots & \ddots \\
s_0 &&& c_{k-1}
\end{pmatrix},
\ \ X_k^{(S4)}(r) = \begin{pmatrix}
c_0 &&& -s_{k-1} \\
&\ddots & \iddots \\
&\iddots & \ddots \\
-s_0 &&& c_{k-1}
\end{pmatrix}.
\end{equation}
Here, $c_\ell=\cos\frac{(1-2r)(2\ell+1)\pi}{4k}$ and
$s_\ell=\sin\frac{(1-2r)(2\ell+1)\pi}{4k}.$
These matrices are used for the so-called
\emph{skew} DCT and DST~\cite{Pueschel:08b}.
Further, let us define $r^{(i)} = (2i+1)/(4m)$ and
$$
r^{(i)}_j = \begin{cases}
\frac{r^{(i)}+2j}{k},& \text{ if $j$ is even} \\
\frac{2-r^{(i)}+2j}{k},& \text{ if $j$ is odd} \\
\end{cases}
$$
for $0\leq j< \floor{\frac{k}{2}}.$
In case $k$ is odd, we also define $r^{(i)}_{k-1}=\frac{r^{(i)}-1}{k}+1$.
Finally, let us define diagonal matrices
\begin{eqnarray*}
D^{(C4)}_k(r^{(i)}) &=& \diag{\Big(1/\cos{(r^{(i)}_j\pi/2)} \Big)}_{0\leq j<k} , \\
D^{(S4)}_k(r^{(i)}) &=& \diag{\Big(\sin{(2kr^{(i)}_j\pi)}/\cos{(r^{(i)}_j\pi/2)} \Big)}_{0\leq j<k} .
\end{eqnarray*}
Then $M = K^n_k \widehat{M} L_{2m}^{n},$ where $\widehat{M} = $
$$
\begin{pmatrix}
\oplus_{i=0}^{m-1} D^{(C4)}_k(r^{(i)}) \DCTIV_k(r^{(i)}) X^{(C4)}_k(r^{(i)})
& \oplus_{i=0}^{m-1} D^{(S4)}_k(r^{(i)}) \DSTIV_k(r^{(i)}) X^{(S4)}_k(r^{(i)})   \\
\oslash_{i=m}^{2m-1} D^{(C4)}_k(r^{(i)}) \DCTIV_k(r^{(i)}) X^{(C4)}_k(r^{(i)})
& \oslash_{i=m}^{2m-1} D^{(S4)}_k(r^{(i)}) \DSTIV_k(r^{(i)}) X^{(S4)}_k(r^{(i)})
\end{pmatrix}.
$$

We can further simplify~\eqref{eq:polDCTIV} by substituting $\pol{\DCTIV}_{2km}$ and $\pol{\DSTIII}_m$
with $\DCTIV_{2km}$ $\DSTIII_m$ using~\eqref{eq:polDCTIV} and~\eqref{eq:polDSTIII}.
Then we use the equalities
\begin{eqnarray*}
X^{(C4)}_k(r)
&=&
X^{(S4)}_k(1-r),\\
\DSTIII_m
&=&
\diag{\Big( (-1)^j\Big)}_{0\leq j< m} \cdot
\DCTIII_m \cdot
J_m,\\
\DSTIV_k
&=&
\diag{\Big( (-1)^j \Big)}_{0\leq j< k} \cdot
\DCTIV_k \cdot
J_k,
\end{eqnarray*}
to obtain the decomposition
\begin{eqnarray*}
\DCTIV_{2km}
&=&
K^{2km}_k (K^{2m}_2 \otimes \DCTIV_k)
Y^{2km}_m (\DCTIII_m \otimes L^{2k}_2) (K^{2km}_{2k})^T \\
&&\cdot
I_k \otimes
\begin{pmatrix}
1\\
&L_2^{2(m-1)}\cdot I_{m-1}\otimes \DFT_2\\
&&1
\end{pmatrix}
(K^{2km}_{2m})^T,
\end{eqnarray*}
where
\begin{equation}
\label{eq:Wang_Y}
Y^{2km}_m =
\bigoplus_{j=0}^{m-1}
\begin{pmatrix}
X^{(C4)}_k(r^{(j)}) &   (-1)^j\cdot J_k \cdot X^{(C4)}_k(1-r^{(j)}) \\
X^{(C4)}_k(1-r^{(j)}) &   (-1)^{j+1}\cdot J_k \cdot X^{(C4)}_k(r^{(j)})
\end{pmatrix}
\end{equation}
and $X^{(C4)}_k(r)$ is defined in~\eqref{eq:X_C4}.


\bibliographystyle{siam}
\begin{small}
\bibliography{references}

\begin{thebibliography}{10}

\bibitem{Auslander:84}
{\sc L.~Auslander, E.~Feig, and S.~Winograd}, {\em {A}belian semi-simple
  algebras and algorithms for the discrete {F}ourier transform}, Advances in
  Applied Mathematics, 5 (1984), pp.~31--55.

\bibitem{Auslander:84a}
\leavevmode\vrule height 2pt depth -1.6pt width 23pt, {\em The multiplicative
  complexity of the discrete {F}ourier transform}, Advances in Applied
  Mathematics, 5 (1984), pp.~87--109.

\bibitem{Bergland:68}
{\sc Glenn~D. Bergland}, {\em Numerical analysis: A fast {Fourier} transform
  algorithm for real-valued series}, Communications ACM, 11 (1968),
  pp.~703--710.

\bibitem{Beth:84}
{\sc Th. Beth}, {\em {V}erfahren der {S}chnellen {F}ouriertransformation
  [Methods for the Fast {F}ourier Transform]}, Teubner, 1984.

\bibitem{Beth:87}
\leavevmode\vrule height 2pt depth -1.6pt width 23pt, {\em {On the
  computational complexity of the general discrete {F}ourier transform}},
  Theoretical Computer Science, 51 (1987), pp.~331--339.

\bibitem{Boyd:2001}
{\sc J.P. Boyd}, {\em Chebyshev and Fourier Spectral Methods}, Dover, 2nd~ed.,
  2001.

\bibitem{Bracewell:84}
{\sc R.~N. Bracewell}, {\em The fast {H}artley transform}, Proc. IEEE, 72
  (1984), pp.~1010--1018.

\bibitem{Britanak:99}
{\sc V.~Britanak and K.~R. Rao}, {\em The fast generalized discrete {F}ourier
  transforms: A unified approach to the discrete sinusoidal transforms
  computation}, Signal Processing, 79 (1999), pp.~135--150.

\bibitem{Clausen:88}
{\sc M.~Clausen}, {\em {B}eitr{\"a}ge zum {E}ntwurf schneller
  {S}pektraltransformationen ({H}abilitationsschrift)}, Univ. Karlsruhe, 1988.

\bibitem{Clausen:93}
{\sc M.~Clausen and U.~Baum}, {\em Fast Fourier Transforms}, BI-Wiss.-Verl.,
  1993.

\bibitem{Cooley:65}
{\sc J.~W. Cooley and J.~W. Tukey}, {\em An algorithm for the machine
  calculation of complex {F}ourier series}, Math. of Computation, 19 (1965),
  pp.~297--301.

\bibitem{Curtis:62}
{\sc W.~C. Curtis and I.~Reiner}, {\em Representation Theory of Finite Groups},
  Interscience, 1962.

\bibitem{Diaconis:90}
{\sc P.~Diaconis and D.~Rockmore}, {\em Efficient computation of the {F}ourier
  transform on finite groups}, Amer. Math. Soc., 3(2) (1990), pp.~297--332.

\bibitem{Driscoll:97}
{\sc J.~R. Driscoll, D.~M. Healy~Jr., and D.~Rockmore}, {\em Fast discrete
  polynomial transforms with applications to data analysis for distance
  transitive graphs}, SIAM Journal Computation, 26 (1997), pp.~1066--1099.

\bibitem{Duhamel:86}
{\sc P.~Duhamel}, {\em Implementation of {"split-radix" FFT} algorithms for
  complex, real, and real-symmetric data}, IEEE Trans.~ASSP, 34 (1986),
  pp.~285--295.

\bibitem{Dummit:03}
{\sc D.~S. Dummit and R.~M. Foote}, {\em Abstract Algebra}, Wiley, 3rd~ed.,
  2003.

\bibitem{Fuhrman:96}
{\sc P.A. Fuhrman}, {\em A Polynomial Approach to Linear Algebra}, Springer
  Verlag, New York, 1996.

\bibitem{Heideman:86}
{\sc M.~T. Heideman and C.~S. Burrus}, {\em On the number of multiplications
  necessary to compute a length-$2^n$ {DFT}}, IEEE Trans. Acoust., Speech,
  Signal Proc., ASSP-34 (1986), pp.~91--95.

\bibitem{Johnson:85}
{\sc H.~W. Johnson and C.~S. Burrus}, {\em On the structure of efficient {DFT}
  algorithmss}, IEEE Trans. Acoust., Speech, Signal Proc., ASSP-33 (1985),
  pp.~248--254.

\bibitem{Kailath:97}
{\sc Th. Kailath and V.~Olshevsky}, {\em Displacement structure approach to
  polynomial {V}andermonde and related matrices}, Linear Algebra and
  Applications, 261 (1997), pp.~49--90.

\bibitem{Mallat:99}
{\sc S.~Mallat}, {\em A Wavelet Tour of Signal Processing}, Academic Press,
  1999.

\bibitem{Mandyam:96}
{\sc G.~Mandyam and N.~Ahmed}, {\em The discrete {L}aguerre transform:
  {D}erivation and applications}, IEEE Trans.~on Signal Processing, 44 (1996),
  pp.~2925--2931.

\bibitem{Martens:90a}
{\sc J.-B. Martens}, {\em The {H}ermite transform---applications}, IEEE
  Trans.~on Acoustics, Speech, and Signal Processing, 38 (1990),
  pp.~1607--1618.

\bibitem{Martens:90}
\leavevmode\vrule height 2pt depth -1.6pt width 23pt, {\em The {H}ermite
  transform---theory}, IEEE Trans.~on Acoustics, Speech, and Signal Processing,
  38 (1990), pp.~1595--1605.

\bibitem{Mason:02}
{\sc J.~C. Mason and D.~C. Handscomb}, {\em {C}hebyshev polynomials}, {C}hapman
  and {H}all/{CRC}, 2002.

\bibitem{Nicholson:71}
{\sc P.J. Nicholson}, {\em Algebraic theory of finite {F}ourier transforms},
  Journal of Computer and System Sciences, 5 (1971), pp.~524--547.

\bibitem{Nussbaumer:82}
{\sc H.~J. Nussbaumer}, {\em Fast {F}ourier Transformation and Convolution
  Algorithms}, Springer, 2nd~ed., 1982.

\bibitem{Potts:98}
{\sc D.~Potts, G.~Steidl, and M.~Tasche}, {\em Fast algorithms for discrete
  polynomial transforms}, Mathematics of Computation, 67 (1998),
  pp.~1577--1590.

\bibitem{Pueschel:05e}
{\sc M.~P{\"u}schel and J.~M.~F. Moura}, {\em Algebraic signal processing
  theory}.
\newblock available at http://arxiv.org/abs/cs.IT/0612077, parts of this
  manuscript appeared as \cite{Pueschel:08} and \cite{Pueschel:08a}.

\bibitem{Pueschel:03a}
\leavevmode\vrule height 2pt depth -1.6pt width 23pt, {\em The algebraic
  approach to the discrete cosine and sine transforms and their fast
  algorithms}, SIAM Journal of Computing, 32 (2003), pp.~1280--1316.

\bibitem{Pueschel:08a}
\leavevmode\vrule height 2pt depth -1.6pt width 23pt, {\em Algebraic signal
  processing theory: {1-D} space}, IEEE Transactions on Signal Processing, 56
  (2008), pp.~3586--3599.

\bibitem{Pueschel:08b}
\leavevmode\vrule height 2pt depth -1.6pt width 23pt, {\em Algebraic signal
  processing theory: {Cooley}-{Tukey} type algorithms for {DCT}s and {DST}s},
  IEEE Transactions on Signal Processing, 56 (2008), pp.~1502--1521.

\bibitem{Pueschel:08}
\leavevmode\vrule height 2pt depth -1.6pt width 23pt, {\em Algebraic signal
  processing theory: Foundation and {1-D} time}, IEEE Transactions on Signal
  Processing, 56 (2008), pp.~3572--3585.

\bibitem{Rao:90}
{\sc K.~R. Rao and P.~Yip}, {\em Discrete Cosine Transform: Algorithms,
  Advantages, Applications}, Academic Press, 1990.

\bibitem{Rockmore:94}
{\sc D.~Rockmore}, {\em Efficient computation of {F}ourier inversion for finite
  groups}, Assoc. Comp. Mach., 41 (1994), pp.~31--66.

\bibitem{Sorensen:85}
{\sc H.~V. Sorensen, D.~L. Jones, C.~S. Burrus, and M.~T. Heideman}, {\em On
  computing the discrete {H}artley transform}, IEEE Trans.~on Acoustics,
  Speech, and Signal Processing, ASSP-33 (1985), pp.~1231--1238.

\bibitem{Steidl:92}
{\sc G.~Steidl}, {\em Fast radix-p discrete cosine transform}, Appl. Algebra
  Engrg. Comm. Comp., 3 (1992), pp.~39--46.

\bibitem{Steidl:91}
{\sc G.~Steidl and M.~Tasche}, {\em A polynomial approach to fast algorithms
  for discrete {F}ourier-cosine and {F}ourier-sine transforms}, Mathematics of
  Computation, 56 (1991), pp.~281--296.

\bibitem{Tolimieri:97}
{\sc R.~Tolimieri, M.~An, and C.~Lu}, {\em Algorithms for Discrete Fourier
  Transforms and Convolution}, Springer, 2nd~ed., 1997.

\bibitem{VanLoan:92}
{\sc C.~Van~Loan}, {\em Computational Framework of the Fast {F}ourier
  Transform}, Siam, 1992.

\bibitem{Voronenko:09}
{\sc Y.~Voronenko and M.~P{\"u}schel}, {\em Algebraic signal processing theory:
  {Cooley}-{Tukey} type algorithms for real {DFT}s}, IEEE Trans. Signal Proc.,
  57 (2009), pp.~205--222.

\bibitem{Wang:84}
{\sc Z.~Wang}, {\em Fast algorithms for the discrete {W} transform and for the
  discrete {F}ourier transform}, IEEE Trans.~on Acoustics, Speech, and Signal
  Processing, ASSP-32 (1984), pp.~803--816.

\bibitem{Winograd:79}
{\sc S.~Winograd}, {\em On the multiplicative complexity of the discrete
  {F}ourier transform}, Advances in Mathematics, 32 (1979), pp.~83--117.

\end{thebibliography}
\end{small}
\end{document}